\documentclass[11pt]{article}
\usepackage{latexsym}
\usepackage{amssymb}
\input QED.sty

\newtheorem{theorem}{Theorem}[section]
\newtheorem{lemma}[theorem]{Lemma}
\newtheorem{proposition}[theorem]{Proposition}
\newtheorem{corollary}[theorem]{Corollary}
\newtheorem{mydefinition}[theorem]{Definition}
\newtheorem{myremark}[theorem]{Remark}
\newtheorem{myexample}[theorem]{Example}
\newtheorem{myfact}[theorem]{Fact}
\newtheorem{myoproblem}[theorem]{Problem}
\newtheorem{myconjecture}[theorem]{Conjecture}
\newenvironment{definition}{\begin{mydefinition}\rm}{\end{mydefinition}}
\newenvironment{remark}{\begin{myremark}\rm}{\end{myremark}}
\newenvironment{example}{\begin{myexample}\rm}{\end{myexample}}

\newenvironment{problem}{\begin{myoproblem}\rm}{\end{myoproblem}}

\newcommand{\CC}{\mathbb C}

\newcommand{\NN}{\mathbb N}
\newcommand{\QQ}{\mathbb Q}
\newcommand{\RR}{\mathbb R}
\newcommand{\ZZ}{\mathbb Z}

\newcommand{\bL}{\underline{L}}
\newcommand{\bVP}{\underline{\mathrm{VP}}}
\newcommand{\chara}{{\mathrm{char}\,}} % \char predefined
\newcommand{\EC}{{\mathit{EC}}}
\newcommand{\grad}{\mathrm{grad}}
\newcommand{\graph}{\mathrm{graph}}
\newcommand{\ini}{\mathrm{in}}
\newcommand{\PER}{\mathrm{PER}}
\newcommand{\PH}{\mathrm{PH}}
\newcommand{\Par}{\mathrm{PAR}}

\newcommand{\Po}{{\mathrm{P}}}      % \P predefined

\newcommand{\res}{{\mathrm{res}}}
\newcommand{\NC}{{\mathrm{NC}}}
\newcommand{\NP}{{\mathrm{NP}}}

\newcommand{\supp}{\mathrm{supp}}
\newcommand{\sus}{\subseteq}       % \ss predefined
\newcommand{\UP}{\mathrm{UP}}
\newcommand{\VNP}{\mathrm{VNP}}
\newcommand{\VP}{\mathrm{VP}}

\begin{document}

\title{The Complexity of Factors of Multivariate Polynomials
\thanks{A preliminary version of this work appeared in  
        Proc.\ 42nd FOCS 2001, pp.~378-385, Oct.~14-17, 2001, Las
        Vegas. The full version has been published at J.~FoCM 4(4):
        369--396, 2004. In this version, we have corrected an error in
        the statement and proof of Theorem~5.7.}}
\author{Peter B\"urgisser\\
Dept.\ of Mathematics and Computer Science\\
University of Paderborn\\
D-33095 Paderborn, Germany\\
buergisser@upb.de}
\date{}

\maketitle
\thispagestyle{empty}

\begin{abstract}
The existence of string functions, which are not polynomial time 
computable, but whose graph is checkable in polynomial time, is a basic 
assumption in crypto\-graphy. 
We prove that in the framework of algebraic complexity, there are no such 
families of polynomial functions of polynomially bounded degree over 
fields of characteristic zero.
The proof relies on a polynomial upper bound on the approximative complexity of 
a factor~$g$ of a polynomial~$f$ in terms of the (approximative) complexity of~$f$
and the degree of the factor~$g$. This extends a result by Kaltofen (STOC 1986). 
The concept of approximative complexity allows to cope with the case that a factor has 
an exponential multiplicity, by using a perturbation argument. 
Our result extends to randomized (two-sided error) decision complexity. 
\end{abstract}

%%%%%%%%%%%%%%%%%%%%%%%%%%%%%%%%%%%%%%%%%%%%%%%%%%%%%%%%%%%%%%%%%%%%%%%%%%%
\section{Introduction}\label{se:intro}

Checking or verifying a solution to a computational problem might be easier 
than computing a solution. In a certain sense, this is the contents 
of the famous $\Po\ne \NP$ hypothesis. In~\cite{vali:76-1} Valiant made an attempt 
to clarify the principal relationship between the complexity of checking and 
evaluating. In particular, he asked whether any (string) function, for which values can be 
checked in polynomial time, can also be evaluated in polynomial time.
Cryptographers hope that the answer to this question is negative, since it turns out to be 
intimately connected to the existence of one-way functions. 
Indeed, the inverse~$\varphi$ of a one-way function is not polynomial time computable, but 
membership to the graph of $\varphi$ can be decided in polynomial time. 
The converse is also known to be true~\cite{grse:88,selm:92} and equivalent to $\Po\ne \UP$.

The goal of this paper is to investigate the relationship between the complexity of 
computational and decisional tasks in an algebraic framework of computation, a line  
of research initiated by Lickteig~\cite{lick:88,lick:90}. 
Unless stated otherwise, $k$ denotes a fixed field of characteristic zero. 
Are there families of polynomials $(\varphi_n)$ over $k$, for which checking the value can be 
done with a polynomial number of arithmetic operations and tests, but which cannot be 
evaluated with a polynomial number of arithmetic operations?
We do not know the answer to this question. 
However, we will be able to show that the answer is negative under the restriction that
the degree of~$\varphi_n$ grows at most polynomially in~$n$.  
Actually, our result is slightly weaker in the sense that we know it to be true 
only for a notion of approximative complexity. 

%%%%%%%%%%%%%%%%%%%%%%%%%%%%%%%%%%%%%%%%%%%%%%%%%%%%%%%%%%%%%%%%%%%%%%%%%%%%
\subsection{Decision, Computation, and Factors} 

We discuss a basic relationship between the complexity of decision and computation 
in our algebraic framework of computation and raise some natural open questions. 

By the {\em straight-line complexity}~$L(g)$ of a multivariate polynomial $g$ over~$k$ 
we understand the minimal number of arithmetic operations sufficient to compute 
$g(X_1,\ldots,X_n)$ by a straight-line program without divisions from the variables~$X_i$ 
and constants in $k$. 
The {\em decision complexity} $C(g)$ of~$g$ is defined as the minimal number of 
arithmetic operations and tests sufficient for an algebraic computation tree to decide 
for given points $x$ in $k^n$ whether $g(x)=0$. 
If $k=\RR$, we allow also $\le$-tests. Clearly, $C(g)\le L(g)+1$
(the $1$ accounts for the zero test). 
We define the {\em exclusion complexity} $\EC(g)$ of $g$ similarly as in \cite{lick:90,bukl:93}
$$
 \EC(g) := \min \{ L(f) \mid f\in k[X_1,\ldots,X_n]\setminus \{0\},\ g\mid f \} .
$$
We clearly have $\EC(g)\le L(g)$.
For formal definitions, we refer to~\cite{bucs:96}.

In the following, we assume that $g$ is the irreducible generator of a hypersurface 
in $k^n$ and either $k=\RR$ or $\CC$.  
Let $f$ be a nonzero polynomial multiple of~$g$, say $f=g^e h$ with a polynomial $h$ 
coprime to $g$ and $e\in\NN_{>0}$.   
Then any straight-line program for $f$ can be used to exclude 
membership to the zeroset of $g$: $f(x)\ne 0$ implies $g(x)\ne 0$
and the converse is also true, provided $h(x)\ne 0$. 
Thus we may consider $\EC(g)$ as a ``generic decision complexity'' of $g$. 

The following well-known lemma provides a link between decisional and 
computational complexity (cf.~\cite{buls:92,bcss:95}). The proof is a rather 
straightforward consequence of the Nullstellensatz.  

\begin{lemma}\label{le:de-comp}
Let $g$ be the irreducible generator of a hypersurface in $\RR^n$ or $\CC^n$. 
Then $\EC(g)\le C(g)$. 
\end{lemma}

Over the reals, we need both assumptions that $g$ is irreducible 
(see the comment on question~(2) below) and that 
the zeroset of $g$ is a hypersurface (take $g=X_1^{2m}+\ldots + X_n^{2m}$ over $\RR$). 
Over the complex numbers, one can relax these assumptions and show that 
$\EC(g) \le C(g) + r-1$ if $g$ is squarefree with $r$ irreducible factors. 
Moreover, we remark that the conclusion of this lemma remains true over any infinite field~$k$ 
if $g$ is the generator of the graph of a polynomial $\varphi$, 
that is, $g =Y-\varphi(X_1,\ldots,X_n)$.

Under the assumption of the lemma we have $\EC(g)\le C(g) \le L(g)+1$. 
Asking about inequalities in the reverse direction, it 
is natural to raise the following questions:
\begin{eqnarray}\label{eq:1}
 L(g) &\le& \EC(g)^{O(1)} \ ?\\
 L(g) &\le& C(g)^{O(1)} \ ? \label{eq:2}\\
 L(\varphi) &\le& C(\graph(\varphi))^{O(1)} \ ? \label{eq:3} 
\end{eqnarray}
Again, $g$ denotes the irreducible generator of a hypersurface in $\RR^n$ or $\CC^n$
and $\varphi$ denotes a polynomial. 
We have the following chain of implications:
$(\ref{eq:1}) \Rightarrow (\ref{eq:2}) \Rightarrow (\ref{eq:3})$.

We believe that all these three questions have negative answers, but we have been  
unable to prove this for irreducible~$g$. 
However, the following counterexamples are known for questions~(\ref{eq:1}) and (\ref{eq:2}) 
when allowing for reducible polynomials~$g$, and assuming $k=\RR$ for question~(\ref{eq:2}).  

Referring to question~(\ref{eq:1}), 
there exist univariate polynomials $f$ having reducible factors $g$ with a complexity exponential 
in the complexity of $f$, a fact first discovered by Lipton and Stockmeyer~\cite{list:78}.
The simplest known example illustrating this is as follows: 
Consider $f_n = X^{2^n} -1 =\prod_{j<2^n} (X-\zeta^j)$, 
where $\zeta = \exp (2\pi i/2^n)$. By repeated squaring we get $L(f_n) \le n+1$. 
On the other hand, one can prove that for almost all $M\sus\{0,1,\ldots,2^n-1\}$ 
the random factor $\prod_{j\in M} (X - \zeta^j)$ has a complexity which is 
exponential in~$n$, cf.~\cite[Exercise~9.8]{bucs:96}. 
A similar reasoning can be made over the reals using Chebychev polyomials.

Commenting on question~(\ref{eq:2}), we remark that the answer is negative if we drop  
the irreducibility assumption and assume $k=\RR$. This follows from the 
following trivial example from~\cite{buls:92}, which shows that 
$L(g)$ may be exponentially larger than $C(g)$:  
Let $g_n\in\RR[X]$ have $n$ distinct real roots. 
Then $C(g_n) \le \log n$ using binary search, but $L(g_n)\ge n$ if the roots of $g$ 
are algebraically independent over $\QQ$. 

We regard to question~(\ref{eq:3}), we note that its truth would imply that there are 
no ``one-way functions'' in the algebraic setting of computations with polynomials. 

It would be interesting to find out whether the truth of the above questions is 
equivalent to the collapse of some complexity classes, similarly as $\Po= \UP$
in the bit-model.  

%%%%%%%%%%%%%%%%%%%%%%%%%%%%%%%%%%%%%%%%%%%%%%%%%%%%%%%%%%%%%%%%%%%%%%%%%%%%
\subsection{Main Results} 

The counterexamples discussed above established polynomials~$g$ whose degree 
was exponential in the exclusion or decision complexity of $g$, respectively. 
We restrict now our attention to factors $g$ having a degree polynomially 
bounded in the complexity of~$f$. 

The {\em Factor Conjecture} from~\cite[Conj.~8.3]{buer:00-3} states that 
for polynomials~$g$ 
\begin{equation}\label{eq:fa-con}
 L(g) \ \le\ (\EC(g) + \deg g)^{O(1)} .
\end{equation}
A partial step towards establishing this conjecture is an older result due to 
Kaltofen~\cite{kalt:86-1}, which can be seen as a byproduct of his 
achievements~\cite{kalt:89-1} to factor polynomials given by a straight-line program 
representations (see also~\cite{katr:90}). 
Kaltofen proved that the complexity of any factor~$g$ of $f$ is 
polynomially bounded in the complexity of~$f$ and in the degree and 
the {\em multiplicity} of the factor~$g$. 

Before stating the precise result, let us fix some notation. For the remainder of this paper, 
$M(d)$ denotes an upper bound on the complexity for the 
multiplication of two univariate polynomials of degree~$d$ over~$k$, that is, 
for computing the coefficients of the product polynomial from the 
coefficients of the given polynomials.
It is well-known that $M(d) \le O(d\log d)$ for $k=\RR$ or $\CC$, cf.~\cite{bucs:96}. 
We will assume that $M(d_1)\le M(d_2)$ for $d_1\le d_2$ and the subadditivity property
$M(d_1) + M(d_2) \le M(d_1 + d_2)$.  

Here is the precise statement of Kaltofen's result, which was independently found by 
the author, compare~\cite[Thm.~8.14]{buer:00-3}. 

\begin{theorem}\label{th:fa-I}
Let $f=g^e h$ with coprime polynomials $g,h\in k[X_1,\ldots,X_n]$. 
Let $d\ge 1$ be the degree of~$g$. We suppose that $k$ is a field of 
characteristic zero. Then we have 
$$
 L(g) \ \le\ O\big( M(d^3e) (L(f) + d\log e) \big)   .
$$
\end{theorem}

Thus our Factor Conjecture claims that the dependence on the multiplicity can be omitted.
It is known~\cite{kalt:86-1} that this is true in the case $f=g^e$, in which case 
$L(g) \le O(M(d)L(f))$ with $d=\deg g$, see Proposition~\ref{pro:root} in the appendix.

The main result of this paper states that the 
dependence on the multiplicity can indeed be omitted when switching to an approximative complexity
measure. The approximative complexity $\bL(g)$ of a polynomial~$g$ is the minimal cost of
``approximative straight-line programs'' computing approxi\-mations of~$g$ with any 
precision required. A formal definition will be given in Section~\ref{se:approx}.

The precise formulation is as follows: 

\begin{theorem}\label{th:main}
Let $k$ be a field of characteristic zero. 
Assume that $n$ by $n$ matrices over~$k$ can be multiplied with 
$O(n^\gamma)$ arithmetic operations in~$k$.
For $g\in k[X_1,\ldots,X_n]$ of degree~$d$ we have 
$$
 \bL(g) \ \le\ O\big(M(d)M(d^4)\EC(g) + d^{2\gamma}M(d)^2\big) .
$$
\end{theorem}

We remark that the ``exponent~$\gamma$ of matrix multiplication'' may be chosen as $2\le\gamma < 2.38$, 
see~\cite{cowi:90,bucs:96}. 

\begin{remark}\label{re:main-supp}
There is certainly room for improvement in this bound. 
In fact, the proof of Theorem~\ref{th:main} yields better estimates in the following cases.
\begin{enumerate}
\item If $g$ is the generator of the graph of a polynomial function $\varphi$, we obtain 
$\bL(g) \le O\big(M(d^2)\EC(g)\big)$. 
\item We have 
$\bL(g) \le O\big(M(d^4)\EC(g) + d^{2\gamma}M(d)\big)$
if $g$ is the irreducible generator of a hypersurface in $\RR^n$ or $\CC^n$. 
\end{enumerate}
\end{remark}

An interesting consequence is the following degree bounded version of question~(\ref{eq:3}):

\begin{corollary}\label{cor:main}
The approximative complexity $\bL(\varphi)$ of a polynomial~$\varphi$ is polynomially 
bounded in the decision complexity of the graph of~$\varphi$ and the degree $d$ of $\varphi$,
namely 
$\bL(\varphi) \le O\big(M(d^2)\,C(\graph(\varphi))\big)$.  
This remains true if we allow randomization with two-sided error. 
\end{corollary}

Coming back to the discussion of one-way functions, we remark that 
Sturtivant and Zhang~\cite{stza:90} obtained the following related result, which excludes the 
existence of certain one-way functions in the algebraic framework of computation. 
Let $\psi\colon k^n\to k^n$ be bijective such that $\psi$ as well as $\psi^{-1}$ are polynomial mappings. 
Then the complexity to evaluate $\psi$ is polynomially bounded in the complexity to evaluate the 
inverse~$\psi^{-1}$ and the maximal degree of the component functions of~$\psi$. 
Again, it is unknown whether the degree restriction can be omitted. 

The paper is organized as follows: 
In Section~\ref{se:approx} we introduce the concept of approximative complexity. 
Section~\ref{se:ac-factors} contains the proof of the main result. 
We then shortly discuss some applications in Section~\ref{se:appl}, where 
we also build in the concept of approximative complexity 
into Valiant's algebraic $\Po$-$\NP$ framework~\cite{vali:79-3,vali:82} 
(see also~\cite{bucs:96,buer:00-3}). 
Section~\ref{se:approx-supp} is devoted to a more detailed analysis of the 
concept of approximative complexity.  
Finally, the appendix contains a proof of Theorem~\ref{th:fa-I}. 

For some other aspects of the issues discussed in this paper see~\cite{buer:01b}.

\bigskip

\noindent{\bf Acknowledgments:} Thanks go to Erich Kaltofen for communicating to me his 
paper~\cite{kalt:86-1} and to an anonymous referee for pointing out the reference~\cite{stza:90}. 
I am grateful to Alan Selman for answering my questions about the complexity of one-way functions. 

\bigskip

\noindent{\bf Note added in proof:} Thomas Lickteig informed me that his unpublished 
papers~\cite[\S 4]{lick:88} and \cite[Thm.~(H.3)]{lick:90} already contain a proof 
of the central result in Section~\ref{sse:graph} (Proposition~\ref{pro:graph}), 
which is based on the same method.

%%%%%%%%%%%%%%%%%%%%%%%%%%%%%%%%%%%%%%%%%%%%%%%%%%%%%%%%%%%%%%%%%%%%%%%%%%%
\section{Approximative Complexity}\label{se:approx}

In complexity theory it has proven useful to study ``approximative algorithms'', 
which use arithmetic with infinite precision
and nevertheless only give us an approximation of the solution to be computed, 
however with any precision required. This concept was systematically studied 
in the framework of bilinear complexity (border rank) and there it has turned 
out to be one of the main keys to the currently best known fast matrix multiplication 
algorithms~\cite{cowi:90}. We refer to~\cite[Chap.~15]{bucs:96} and the references 
there for further information. 

Although approximative complexity is a very natural concept, it has 
been investigated in less detail for computations of polynomials or rational functions. 
Originally, it had been introduced by Strassen in a topological way~\cite{stra:74-1}. 
Griesser~\cite{grie:86-1} generalized most of the known lower bounds for 
multiplicative complexity to approximative complexity. 
Lickteig systematically studied the notion of approximative complexity with 
the goal of proving lower bounds~\cite{lick:90}. 
In Grigoriev and Karpinski~\cite{grka:97} the notion of approximative complexity 
is also employed for proving lower bounds. 

It is not known how to meaningfully relate the complexity of trailing coefficients or 
of factors of a polynomial to the complexity of the polynomial itself. 
However, by allowing approximative computations, we are able to establish 
quite satisfactory reductions in these cases. The deeper reason why this is possible 
seems to be the lower semicontinuity of the approximative complexity, which allows
a controlled passage to the limit and can be used in perturbation arguments. 

Assume the polynomial $f$ is expanded with respect to~$Y$:
$$
 f = f_q(X_1,\ldots,X_n) Y^q + f_{q+1}(X_1,\ldots,X_n)Y^{q+1} + \ldots   .
$$ 
We do not know whether the complexity of the trailing coefficient~$f_q$ can be 
polynomially bounded in the the complexity of $f$.
However, we can make the following observation. For the moment assume that~$k$ 
is the field of real or complex numbers. We have
$\lim_{y\to 0} y^{-q} f(X,y) = f_q(X)$ and 
$L(f(X,y)) \le L(f)$ for all $y\in k$. 
Thus we can approximate $f_q$ with arbitrary precision by polynomials 
having complexity at most~$L(f)$. 
We will say that $f_q$ has ``approximate complexity'' at most $L(f)$.

In what follows, we will formalize this in an algebraic way; a topological 
interpretation will be given later. 
Throughout the paper, $K:=k(\epsilon)$ is a rational function field in the 
indeterminate $\epsilon$ over the field~$k$ and $R$ denotes the local subring of $K$ 
consisting of the rational functions defined at $\epsilon=0$. 
We write $F_{\epsilon=0}$ for the image of $F\in R[X]$ 
under the morphism $R[X]\to k[X]$ induced by $\epsilon\mapsto 0$. 

\begin{definition}\label{def:bL}
Let $f\in k[X_1,\ldots,X_n]$. 
The {\em approximative complexity} $\bL(f)$ 
of the polynomial $f$ is the smallest natural number $r$ such that there exists 
$F$ in $R[X_1,\ldots,X_n]$ satisfying 
$F_{\epsilon=0} = f$ and $L(F) \le r$. 
Here the complexity~$L$ is to be interpreted with respect to the larger 
field of constants~$K$. 
\end{definition}

Even though $L$ refers to division-free straight-line programs, 
divisions may occur implicitly since our model allows the free use of 
any elements of~$K$ as constants (e.g., division by powers of $\epsilon$). 
In fact, the point is that even though $F$ is defined with respect to the 
morphism $\epsilon\mapsto 0$, the intermediate results of the computation may not be so!  
Note that $\bL(f)\le L(f)$. 

We remark that the assumption that any elements of $K$ are free constants is just 
made for conceptual simplicity. 
We may as well require to build up the needed elements of~$K$ from $\epsilon, \epsilon^{-1}$ and 
elements of~$k$. It is easy to see that this would not change our main 
result (i.e., Theorem~\ref{th:main}).

Assume that $\bL(f)\le r$ over $k=\RR$, say 
$F_\epsilon(x) = f(x) + \epsilon R_\epsilon(x)$ and $L(F_\epsilon)\le r$. 
Let $S$ be the supremum of $R_\epsilon(x)$ over all $\epsilon\in [0,1]$ and 
$x\in\RR^n$ with $||x||_\infty \le 1$. 
Then we have for such $\epsilon$ and $x$ that 
$ | F_\epsilon(x) - f(x) | = \epsilon\, |R_\epsilon(x)| \le S\epsilon$.
Therefore, for each $\epsilon >0$ we can compute on input $x$ an approximation 
to $f(x)$ with absolute error less than $S\epsilon$ with only~$r$ arithmetic operations.
If we would additionally require in the definition of $\bL$ to build up the needed constants 
in~$K$ from $\epsilon, \epsilon^{-1}$, then $\bL(f)\le r$ would even mean that one can compute 
an approximation with error less than $S\epsilon$ with only $r$ arithmetic operations 
on input $x$ {\em and $\epsilon$.}

\begin{example}
Let us illustrate the notion of approximative complexity with an example.
The convex hull of the support $\supp f$ of a polynomial 
$$
  f = \sum_{a\in\supp f} c_a X_1^{a_1}\cdots X_n^{a_n} \hspace{15mm} (c_a\ne 0)
$$
is called the Newton polytope $P$ of $f$. To a supporting hyperplane $H$ of $P$ 
we may assign the corresponding initial term polynomial 
$$
 \ini_H f := \sum_{a\in H\cap\supp f} c_a X_1^{a_1}\cdots X_n^{a_n} . 
$$
We claim that 
$$
  \bL(\ini_H f) \le L(f) + n + 1   .
$$
Indeed, we may obtain $\ini_H(f)$ as a ``degeneration'' of $f$ as follows. 
Assume that $\langle w,x \rangle - c=0$ is the equation of $H$, say
$\langle w,x\rangle \ge c$ on $P$. We can always achieve that 
$w\in\ZZ^n$, $c\in\ZZ$. Then we have 
$$
  F := \epsilon^{-c} f(\epsilon^{w_1}X_1,\ldots,\epsilon^{w_n}X_n) = 
  \sum_{a\in \supp f} c_a \epsilon^{\langle w,a\rangle -c} X_1^{a_1}\cdots X_n^{a_n} = \ini_H f + O(\epsilon) 
$$
using the convenient, intuitive Big-Oh notation. 
Therefore, $F_{\epsilon =0} = \ini_H f$ and 
$L(F) \le L(f) + n + 1$ which proves our claim. 
(Recall that the powers of $\epsilon$ are considered as constants.) 
\end{example}

The next lemma states some of the basic properties of $\bL$. 

\begin{lemma}\label{le:aux}
\begin{enumerate}
\item[{\rm 1.}] {\em (Semicontinuity)}  If $F$ is defined over $R$ and 
$f=F_{\epsilon=0}$, then $\bL(f) \le \bL(F)$. 
Note that the quantity $\bL(F)$ is well-defined for a polynomial~$F$ over $K$ 
(adjoining a further indeterminate to $K$).

\item[{\rm 2.}] {\em (Elimination of constants)}\ Let $k_1$ be a field extension of $k$ of degree at most $d$ 
and $f$ be a polynomial over $k$. Then $\bL(f) \le O(M(d)\, \bL_{\,k_1}(f))$, where 
$\bL_{\, k_1}(f)$ denotes the approximative complexity of $f$ interpreted as 
a polynomial over $k_1$ (i.e., constants in $k_1$ may be used freely). 

\item[{\rm 3.}] {\em (Transitivity)}\ The approximative complexity $\bL(f\mid g)$ to compute $f$ from $g$ and 
the variables is defined in a natural way. We have 
$\bL(f) \le \bL(f\mid g) + \bL(g)$, 
and an analogous inequality is true for 
the computation of several polynomials. 

\end{enumerate}
\end{lemma}

\begin{proof}
(1) We start with a general observation: Let $\Phi$ be a rational function 
in two variables $\epsilon,\delta$. We assume that $\Phi$, viewed as a rational function 
in~$\delta$ over $k(\epsilon)$, is defined at $\delta=0$ with value $\Phi_{\delta=0}$. 
Moreover, we assume that the rational function~$\Phi_{\delta=0}$ is defined at $\epsilon=0$ with 
value~$\lambda:=(\Phi_{\delta=0})_{\epsilon=0}$. 
Then $\Phi(\epsilon,\epsilon^N)$ is defined at $\epsilon=0$ with value~$\lambda$ for sufficiently 
large~$N$. 
Indeed, if $B$ is the denominator of $\Phi$ and 
$B(\epsilon,0) = B_{m}\,\epsilon^m + O(\epsilon^{m+1})$, $B_{m}\ne 0$, then 
it is easy to check that it suffices to take $N>m$. 

Let now $\Phi\in k(\epsilon,\delta)[X]$ be such that 
$\Phi_{\delta=0}=F$ and $L(\Phi) = \bL(F)$. 
An optimal computation of $\Phi$ takes place in a finitely generated subring 
of $k(\epsilon,\delta)[X]$. 
The morphism $\delta\mapsto\epsilon^N$ is defined on this subring 
if $N$ is chosen sufficiently large. 
Then we have $L(\phi)\le L(\Phi)$ for $\phi:= \Phi(\epsilon,\epsilon^N,X)$. 
If $N$ is chosen sufficiently large, we have $\phi_{\epsilon=0} = F_{\epsilon=0}$
by the observation at the beginning of the proof. 
This implies the claim.

\medskip

(2) This follows easily from~\cite[Prop.~4.1(iii)]{buer:00-3}. 

\medskip

(3) By definition there exists $G\in k(\epsilon)[X]$ such that $G_{\epsilon=0} = g$ and 
$\bL(g) = L(G)$. Moreover, there exists 
$F\in k(\epsilon)[X]$ such that $F_{\epsilon=0} = f$ and 
$$
  \bL(f \mid g) = L(F \mid g) . 
$$
Let $\Gamma$ be an optimal straight-line program computing $F$ 
from $g$, variables $X_i$, and constants in $k(\epsilon)$.
We replace $\epsilon$ by a new indeterminate $\delta$ and denote the element thus 
corresponding to $G$ by $G(\delta) \in k(\delta)[X]$ (abusing notation). 
If we replace the input $g$ by $G(\delta)$, then the program $\Gamma$,
using the same constants in $k(\epsilon)$ as $\Gamma$, 
will compute an element $\Phi\in k(\epsilon,\delta)[X]$.
Clearly, $\Phi_{\delta=0}=F$. 

Since the computation of $\Phi$ takes place in a finitely generated subring 
of $k(\epsilon,\delta)[X]$, the morphism $\delta\mapsto\epsilon^N$ is defined on this subring 
if $N$ is chosen large enough. If we denote the image of~$\Phi$ under this morphism 
by $\phi$ and the image of $G(\delta)$ by $G(\epsilon^N)$, then we have 
$$
  L(\phi \mid G(\epsilon^N)) \ \le\ L(F\mid g) \ =\ \bL(f \mid g )  .
$$
Moreover, we clearly have 
$L(G(\epsilon^N)) \le L(G) = \bL(g)$.
By the transitivity of~$L$, we get 
$L(\phi) \le \bL(f \mid g) + \bL(g)$. 
From the observation at the beginning of the proof of part~(1) of the lemma, we conclude that 
$\phi_{\epsilon=0} = f$ for sufficiently large~$N$. 
This implies the claim.
\end{proof}

We proceed with a topological interpretation of approximative complexity, 
which points out the naturality of this notion from a mathematical point of view. 
It will not be needed for the proof of the main Theorem~\ref{th:main}. 

Assume $k$ to be an algebraically closed field. 
There is a natural way to put a Zariski topology on the polynomial ring 
$A_n := k[X_1,\ldots,X_n]$ as a limit of the Zariski topologies on the 
finite dimensional subspaces $\{f\in A_n\mid \deg f\le d\}$ for $d\in\NN$. 
If $k$ is the field of complex numbers, we may define the Euclidean topology 
on $A_n$ in a similar way. 

If $f\in A_n$ satisfies $\bL(f) \le r$, then it easy to see that $f$ lies in the closure 
(Zariski or Euclidean) of the set $\{f\in A_n\mid L(f)\le r\}$. 
Indeed, we have $L(F_{\epsilon=y}) \le L(F)$ for all but finitely many $y\in k$ 
and $\lim_{y\to 0} F_{\epsilon=y}=F_{\epsilon=0}=f$.
Alder~\cite{alde:84} has shown that the converse is true and 
obtained the following topological characterization of the approximative complexity. 

\begin{theorem}\label{th:TA}
Let $k$ be algebraically closed. 
The set $\{f\in A_n\mid \bL(f)\le r\}$ is the closure of the set  
$\{f\in A_n\mid L(f)\le r\}$ for the Zariski topology. 
If $k=\CC$, this is also true for the Euclidean topology. 
\end{theorem}

This essentially claims that $\bL$ is the largest lower 
semicontinuous function of~$f$ bounded by $L(f)$. 
The proof of Theorem~\ref{th:OA} in Section~\ref{sse:bL-char}
implies the above result as a special case. 
We remark that this theorem can also be easily deduced 
from~\cite[Lemma~20.28]{bucs:96}. 
One can show that the above statement is also true over the reals with the Euclidean topology, 
similar as in Lehmkuhl and Lickteig~\cite{leli:89}. 

%%%%%%%%%%%%%%%%%%%%%%%%%%%%%%%%%%%%%%%%%%%%%%%%%%%%%%%%%%%%%%%%%%%%%%%%%%%
\section{Approximative Complexity of Factors}\label{se:ac-factors}

We will supply here the proof of our main result Theorem~\ref{th:main}.
The outline of the proof is as follows:
Let $f=g^e h$ with coprime $g$ and $h$ and assume $k=\CC$. 
After a suitable coordinate transformation one can interpret the zeroset of the factor~$g$
locally as the graph of some analytic function $\varphi$. 
In order to cope with a possibly large multiplicity~$e$ of $g$, we 
apply a small perturbation to the polynomial $f$ without affecting its complexity too much.
This results in a small perturbation of $\varphi$. We compute now the homogeneous parts of 
the perturbed $\varphi$ by a Newton iteration up to a certain order. 
Using efficient polynomial arithmetic, this gives us an upper 
bound on the approximative complexity of the homogeneous parts of $\varphi$ up to 
a predefined order (Proposition~\ref{pro:graph}). 
In the special case, where the factor~$g$ is the generator of the 
graph of a polynomial function, we are already done. 
This is essentially the contents of Section~\ref{sse:graph}. 

In a second step, elaborated in Section~\ref{sse:minpol}, 
we view the factor $g$ as the minimal polynomial  of $\varphi$ in $Y:=X_n$ over the 
field $k(X_1,\ldots,X_{n-1})$. We show that the Taylor approximations up to order $2d^2$ 
uniquely determine the factor~$g$ and compute the bihomogeneous components of $g$ 
with respect to the degrees in the $X$-variables and~$Y$ by fast linear algebra. 

%%%%%%%%%%%%%%%%%%%%%%%%%%%%%%%%%%%%%%%%%%%%%%%%%%%%%%%%%%%%%%%%%%%%%%%%%%%
\subsection{Preliminaries}\label{se:aux}

The following result is obtained by a straightforward application of a technique 
introduced by Strassen~\cite{stra:73-1} for the computation of homogeneous components 
and avoiding divisions. A proof will be sketched in Section~\ref{sse:p-adic-coeff}.  

\begin{proposition}\label{pro:comp_hom_parts} 
Assume that $F(X,Y)=\sum_{i,\delta} F_i^{(\delta)}Y^i$ is the bihomogeneous decomposition of 
the polynomial $F\in k[X_1,\ldots,X_n,Y]$ with respect to the total degree in the $X$-variables and the 
degree~$Y$. 
Thus $F_i^{(\delta)}$ is a homogeneous polynomial in the $X$-variables 
of degree~$\delta$. Then we have for all $D\ge 1$
$$
L(\{ F_i^{(\delta)} \mid i,\delta \le D\}) \ \le\ O(M(D^2)L(F)) 
$$
and the same is true if the complexity~$L$ is replaced by 
the approximative complexity $\bL$.
\end{proposition}

Part~(1) of the next lemma follows immediately from the well-known algorithms for 
the multiplication and division of univariate power series described 
in~\cite[\S 2.4]{bucs:96} by 
interpreting the homogeneous components of a multivariate power series $f\in k[[X_1,\ldots,X_n]]$ 
as the $T$-adic coefficients of the transformed series $f(T X_1,\ldots, T X_n)$. 
Part~(2) of this lemma is obtained from part~(1) by applying Horner's rule. 

\begin{lemma}\label{le:compos}
\begin{enumerate}

\item[\rm (1)] We can compute the homogeneous parts up to degree~$D$ of the 
product $F\cdot G$ and of the quotient $F/G$ (if $G(0)\ne 0$) 
of multivariate power series $F$ and $G$ from the homogeneous parts of $F$ and $G$ up to degree~$D$ 
by $M(D)$ arithmetic operations.

\item[\rm (2)] Assume that the multivariate power series $F_0,\ldots,F_D$ and $\Phi$ 
are given by their homogeneous parts up to degree~$D$. 
Then we can compute from this data the homogeneous parts of 
$\sum_{i=0}^D F_i \Phi^i$ up to degree~$D$ by $O(D M(D))$ arithmetic operations. 

\end{enumerate}
\end{lemma}

We remark that $D^2$ nonscalar operations are needed for the composition problem~(2) 
in the generic case.
For proving this, we assume that we have just one variable and choose for $\Phi$ a constant 
power series: $\Phi=a$. Let $F_i = \sum_j F_{i,j} X^j$.
The problem then reduces to the simultaneous evaluation of $\sum_{i\le D} F_{i,j}a^i$ 
for $j\le D$, a problem known to be of nonscalar complexity $(D+1)^2 -1$,  
see~\cite[Exercise~6.2]{bucs:96}. 

%%%%%%%%%%%%%%%%%%%%%%%%%%%%%%%%%%%%%%%%%%%%%%%%%%%%%%%%%%%%%%%%%%%%%%%%%%%
\subsection{Approximative Computation of Graph}\label{sse:graph}

We need the following lemma. 

\begin{lemma}\label{le:fe}
Let $g,h\in k[X_1,\ldots,X_{n}]$ be coprime, $g$ irreducible and $d:=\deg g$. 
Then there is field extension $k_1$ of degree at most~$d$ over $k$ and a point 
$p\in k_1^{n}$ such that 
$$
 g(p)=0, \ h(p)\ne 0, \ \grad\, g(p) \ne 0 .
$$
Moreover, we may assume in this statement that $k_1=k$ if either $k=\CC$
 or if $k=\RR$ and $g$ is the irreducible generator of a hypersurface in $\RR^n$. 
\end{lemma}

\begin{proof}
The claim for $k=\CC$ is a straightforward consequence of the Nullstellensatz. 
In the case $k=\RR$ we apply Theorem~4.5.1 in~\cite{bocr:87}, which tells us that  
$\{x\in\RR^n \mid g(x)=0, \grad\, g(x) \ne 0\}$ is Zariski dense in the zeroset 
of~$g$ and that the vanishing ideal of this zeroset is generated by~$g$. This 
implies the claim. 

In the general case, we apply a linear coordinate transformation 
$X_i \mapsto X_i + u_i Y\ (i< n)$, $Y\mapsto vY$ for suitable $u_i,v\in k$
in order to achieve that $g$~is monic of degree~$d$ with respect to the variable $Y:=X_n$. 
From now on we write $X:=(X_1,\ldots,X_{n-1})$. 
Since $g$ is irreducible and $g,h$ are coprime, the resultants
$\res_Y(g,\partial_Y g)$ and $\res_Y(g,h)$ in $k[X]$
with respect to the variable $Y$ are not the zero polynomials. 
We choose a point $\xi\in k^{n-1}$ where these resultants do not vanish. 
From the properties of the resultant we conclude that the univariate 
polynomials $\tilde{g}:=g(\xi,Y)$ and $\tilde{h}:=h(\xi,Y)$ are coprime 
and that $\tilde{g}$ is squarefree. 
Let $\eta$ be a root of $\tilde{g}$ is some extension field~$k_1$ of $k$ 
of degree at most $d$. Then $\partial_Y\tilde{g}$ and $\tilde{h}$ do not 
vanish at $\eta$ and the point $p=(\xi,\eta)$ satisfies the claim of the lemma.  
\end{proof}

We assume now that we are in the situation of Theorem~\ref{th:main}.
Without loss of generality we may assume that $g$ is irreducible 
(apply Theorem~\ref{th:main} to the irreducible factors of $g$ and 
use the subadditivity and monotonicity of $M$). 
From now on we use the notations $Y:=X_n$ and $X:=(X_1,\ldots,X_{n-1})$. 

Let $f=g^e h$, where $g$ and $h$ coprime such that $\EC(g)=L(f)$. 
We choose the field extension $k_1$ and the point $p=(\xi,\eta)\in k_1^n$ according to Lemma~\ref{le:fe}. 
To simplify notation, we assume that $k_1=k$, an assumption which will 
be eliminated at the end of Section~\ref{sse:minpol} 
at the price of an additional factor $M(d)$ in the complexity bound. 

We are now going to transform the polynomials into a special form by suitable linear transformations.
By a coordinate shift we can always achieve that $(\xi,\eta)=(0,0)$. 
By a substitution $\tilde{g}(X,Y) := g(X_1 + u_1 Y,\ldots,X_{n-1} + u_{n-1}Y, v Y)$ 
we may achieve that the degree of $\tilde{g}$ in $Y$ equals~$d$ and that 
$\partial_Y\tilde{g}(0,0)$ does not vanish. Indeed, if  $g^{(d)}$ denotes the 
homogeneous component of $g$ of degree~$d$, then the coefficient of~$Y^d$ in 
$\tilde{g}$ equals $\tilde{g}^{(d)}(0,1)= g^{(d)}(u,v)$. Moreover, 
$
 \partial_Y\tilde{g}(0,0) = u_1\partial_{X_1}g(0,0) + \ldots + u_{n-1}\partial_{X_{n-1}}g(0,0) + v\partial_Y g(0,0) 
$.
Hence it suffices to choose $u,v$ such that this linear combination does not vanish and such that $g^{(d)}(u,v)\ne 0$.
By scaling, we may assume without loss of generality that $\tilde{g}$ is monic with respect to $Y$.  
In the following, we will assume that this transformation has already been done, i.e., $\tilde{g}=g$, which results in 
a complexity increase of $f$ of at most $2n$. 
Note that $L(f)\ge n$ if all the variables occur in~$f$. 

Summarizing, we achieved the following by a suitable choice of a linear transformation: 
\begin{equation}\label{eq:ass_g}
 g(0,0)=0,\ h(0,0)\, \partial_Yg(0,0) \ne 0,\ \deg_Y g = d .
\end{equation}

The implicit function theorem implies that there exists a unique 
formal power series $\varphi\in k[[X]]$ such that 
\begin{equation}\label{eq:c_phi}
 g(X,\varphi(X))=0,\ \varphi(0)=0 .
\end{equation}
Moreover, this power series can be recursively computed by the following  
{\em Newton iteration}: if we put $\varphi_0 =0$ and define
\begin{equation}\label{eq:newton_g}
 \varphi_{\nu+1} = \varphi_\nu - \frac{g(X,\varphi_\nu)}{\partial_Yg(X,\varphi_\nu)} ,
\end{equation}
then we have quadratic convergence of the $\varphi_\nu$ towards $\varphi$, in the sense 
that $\varphi_\nu\equiv\varphi\bmod (X)^{2^\nu}$, where $(X)$ denotes the maximal ideal 
of $k[[X]]$ (cf.~\cite[Theorem 2.31]{bucs:96}). 

It is easy to see that if the partial derivative 
$\partial_Y f(0,0)$ would not vanish,  then the above power series $\varphi$ could also be 
recursively computed by the Newton recursion~(\ref{eq:newton_g}) with $g$ replaced by $f$. 
However, $\partial_Y f(0,0) = 0$ always vanishes for multiplicities $e >1$. 
The key idea is now to enforce the nonvanishing of this partial derivative by a suitable
perturbation of the given polynomial~$f$. By doing so, we have to content 
ourselves with an approximative computation of the factor~$g$. 

Based on these ideas, we prove the following assuming the conditions~{\rm (\ref{eq:ass_g})}: 

\begin{proposition}\label{pro:graph} 
The homogeneous parts $\varphi^{(\delta)}$ of $\varphi$ of degree~$\delta$ satisfy
$$
 \forall D\ge 1 :\quad \bL(\varphi^{(1)},\ldots,\varphi^{(D)}) \ \le\  O\big(M(D^2)\bL(f)\big) .
$$
\end{proposition}

\begin{proof}
Note that $g$, viewed as a polynomial in $Y$ over $k(X)$, is 
the minimal polynomial of $\varphi$ over $k(X)$. 
W.l.o.g.\ we may assume that $\varphi$ is not a rational function
(otherwise $d=1$, $\varphi$ would be linear, and the claim obvious). 

We define the perturbed polynomial
$F(X,Y):=f(X,Y+\epsilon) - f(0,\epsilon)$ 
over the coefficient ring $R$. 
It is clear that $F(0,0)=0$ and $F_{\epsilon=0}=f$. 
By a straight-forward calculation we get 
$$
 \partial_Y F(0,0)  = 
  \big(e h\, \partial_Y g + g\, \partial_Y h \big)(0,\epsilon)\cdot g^{e-1}(0,\epsilon) .
$$
Assumptions~(\ref{eq:ass_g}) tell us that 
$g(0,\epsilon) = \lambda \epsilon + O(\epsilon^{2})$ with 
$\lambda\in k^\times$, hence 
$$
 \partial_Y F(0,0) = e \, \lambda^e \,h(0,0)\, \epsilon^{e-1} + O(\epsilon^{e}) 
$$
and we conclude that this partial derivative does not vanish ($\chara k= 0$). 

As in the reasoning before, the implicit function theorem implies that there exists 
a unique formal power series $\Phi$ over the field $K=k(\epsilon)$ such that 
$
 F(X,\Phi(X))=0,\ \Phi(0)=0 
$
and this power series can be recursively computed by the Newton iteration
\begin{equation}\label{eq:newton_F}
 \Phi_0 =0,\ \Phi_{\nu+1} = \Phi_\nu - \frac{F(X,\Phi_\nu)}{\partial_Y F(X,\Phi_\nu)} 
\end{equation}
with quadratic convergence: $\Phi_\nu\equiv\Phi\bmod (X)^{2^\nu}$.

\medskip

\noindent {\bf Claim:} $\Phi_\nu$ is defined over the coefficient ring~$R$ for all $\nu$. 

\medskip

We prove this claim by induction on $\nu$, the induction start $\nu=0$ being clear. 
So let us assume that $\Phi_\nu$ is defined over $R$ and set 
$\psi_\nu := (\Phi_\nu)_{\epsilon=0}$.
By applying the morphism $R[[X]]\to k[[X]],\epsilon\mapsto 0$ we obtain
\begin{eqnarray*}
\lefteqn{ (\partial_Y F(X,\Phi_\nu))_{\epsilon=0} = \partial_Y f (X, \psi_\nu) }\\
   & & =  \big( e h\,\partial_Y g +  g\, \partial_Y h \big)(X,\psi_\nu) \cdot g^{e-1} (X, \psi_\nu) .
 \end{eqnarray*}
The first parenthesis  maps under the substitution $X\mapsto 0$ to 
$(e h \,\partial_Y g)(0,0)$, which is nonzero by our assumptions. 
The second factor $g (X, \psi_\nu)$ can only vanish if $\psi_\nu = \varphi$
since the power series $\varphi$ is uniquely determined by the conditions~(\ref{eq:c_phi}). 
In this case, $\varphi$ would be a rational function, which we have excluded at the beginning of the proof. 
We have thus shown that $(\partial_Y F(X,\Phi_\nu))_{\epsilon=0}$ nonzero. 
By equation~(\ref{eq:newton_F}) this implies that 
$\Phi_{\nu+1}$ is defined over $R$ and proves the claim.

%\medskip

The claim implies that $\Phi$ is defined over $R$. 
From $F(X,\Phi(X))=0$ we get 
$f(X,(\Phi(X))_{\epsilon=0}) =0$, hence 
$g(X,(\Phi(X))_{\epsilon=0}) =0$, as $h(X,(\Phi(X))_{\epsilon=0})\ne 0$.
We conclude that $(\Phi(X))_{\epsilon=0} = \varphi$. 
If $\Phi^{(\delta)}$ denotes the homogeneous part of $\Phi$ of degree~$\delta$, 
we have $(\Phi_\nu)^{(\delta)} = \Phi^{(\delta)}$ for  $\delta < 2^\nu$. 
This implies for $\delta < 2^\nu$ that 
$$
 ( (\Phi_\nu)^{(\delta)} )_{\epsilon=0}  =  ( \Phi^{(\delta)} )_{\epsilon=0}  
  = (\Phi_{\epsilon=0})^{(\delta)} = \varphi^{(\delta)} . 
$$

As a word of warning, we point out that a certain care in these argumentations 
is necessary. For instance, Example~\ref{ex:warning} below shows that in general
$(\Phi_\nu)_{\epsilon=0} \ne \varphi_\nu$. 

We turn now to the algorithmic analysis of the proof. 
First of all we note that $L(F)\le L(f)+2$. 
A moment's thought shows that also 
$\bL(F)\le \bL(f)+2$. 
In order to prove the proposition it is enough to show that 
\begin{equation}\label{eq:ass}
 L(\Phi_N^{(1)},\ldots,\Phi_N^{(D)}) \ \le\  O\big(M(D^2)\bL(f)\big) ,
\end{equation}
where $N:= \lceil \log (D+1) \rceil$. 
In fact, by the semicontinuity of $\bL$ (Lemma~\ref{le:aux}(1)), 
we only need to prove this estimate for 
approximative complexity on the lefthand side. 

The following computation deals with polynomials in the $X$-variables, which are truncated at a certain degree 
and represented by their homogeneous parts up to this degree. 
We obtain from Proposition~\ref{pro:comp_hom_parts} for the bihomogeneous decomposition of~$F$ that 
\begin{equation}\label{eq:bihom}
\bL(\{ F_i^{(\delta)} \mid i,\delta \le D\}) \ \le\ O\big(M(D^2)\bL(F)\big) .
\end{equation}
In the following, we assume that we have already computed the 
bihomogeneous components $F_i^{(\delta)}$ for $i,\delta \le D$. 

Inductively, we suppose that we have computed the homogeneous parts of $\Phi_\nu$ up to degree $2^\nu$. 
The main work of one Newton step~(\ref{eq:newton_F}) consists in the computation of the substituted 
polynomials $F(X,\Phi_\nu)$ and $\partial_Y F(X,\Phi_\nu)$. 
By Lemma~\ref{le:compos} we can compute the homogeneous parts up to degree $2^{\nu+1}$ of $F(X,\Phi_\nu)$  
by $O(2^\nu M(2^\nu))$ arithmetic operations.
Analogously, we get the  homogeneous parts up to degree $2^{\nu+1}$ of $\partial_Y F(X,\Phi_\nu)$ by 
the same number of arithmetic operations. 
By a division and a subtraction we obtain from this the homogeneous parts of $\Phi_{\nu+1}$ up to degree $2^{\nu+1}$
using further $O(M(2^\nu))$ arithmetic operations. Altogether, we obtain
\begin{eqnarray*}
 \lefteqn{L(\Phi_N^{(1)},\ldots,\Phi_N^{(D)}\mid \{ F_i^{(\delta)} \mid i,\delta \le D\})}\\
   & \le & O(\sum_{\nu=0}^N 2^\nu M(2^\nu)) \ \le\ O(D M(D))\ \le\ O(M(D^2)) ,
\end{eqnarray*}
by the monotonicity and subadditivity of $M$.
The assertion~(\ref{eq:ass}) follows from this estimate and equation~(\ref{eq:bihom}) 
by the transitivity of approximative complexity (Lemma~\ref{le:aux}(3)). 
\end{proof}

\begin{example}\label{ex:warning}
Consider the bivariate polynomial $g:= (1+Y)^2 - 1 - X^2$ and put $f=g^2$, $h=1$.  
Then the conditions~(\ref{eq:ass_g}) are satisfied. The first Newton iterate 
according to~(\ref{eq:newton_g}) satisfies $\varphi_1 =\frac{1}{2}X^2$ and  
the power series $\varphi$ defined by~(\ref{eq:c_phi}) has the expansion
$$
 \varphi = \frac{1}{2}X^2 -\frac{1}{8}X^4 + \cdots .
$$
As in the proof of Proposition~\ref{pro:graph} we set 
$F:= f(X,Y+\epsilon)-f(0,\epsilon)$. A straightforward computation (e.g., using
a computer algebra system) yields for the first Newton approximation $\Phi_1$ 
according to~(\ref{eq:newton_F}) that
$$
 \Phi_1 = -\frac{1}{4} \frac{(-X^2 + 2\epsilon + \epsilon^2)^2 - \epsilon^2(2+\epsilon)^2}{(-X^2+2\epsilon+\epsilon^2)(1+\epsilon)} .
$$
Therefore, $(\Phi_1)_{\epsilon=0} = \frac{1}{4}X^2 \ne \varphi_1$. On the other hand, we note that 
the expansion of $\Phi_1$ starts as follows
$$
 \Phi_1 = \frac{1}{2(1+\epsilon)}X^2 + \frac{1}{4 \epsilon(1+\epsilon)(2+\epsilon)}X^4 + \cdots 
$$
and we see that $(\Phi_1^{(2)})_{\epsilon=0} =\frac{1}{2}X^2 = \varphi_1^{(2)}$. 
Note that the fourth order term of this expansion is not defined for $\epsilon=0$ even though
$\Phi_1$ is defined under this substitution!
\end{example}

%%%%%%%%%%%%%%%%%%%%%%%%%%%%%%%%%%%%%%%%%%%%%%%%%%%%%%%%%%%%%%%%%%%%%%%%%%%
\subsection{Reconstruction of Minimal Polynomial}\label{sse:minpol}

Consider the bihomogeneous decomposition $g(X,Y)=\sum_{i,\alpha\le d}g_i^{(\alpha)} Y^i$.
Let $T$ be an additional indeterminate and perform the substitution 
$X_j\mapsto T X_j$. The condition $g(X,\varphi(X))= 0 \bmod (X)^{D+1}$ then translates to 
$$
 \sum_{i,\alpha \le d} g_i^{(\alpha)} T^\alpha (\sum_{\delta\le D}\varphi^{(\delta)} 
   T^\delta)^i \ \equiv\ 0 \bmod T^{D+1}  
$$
for any $D\ge 1$. Moreover, we have $g_d^{(0)} =1$ and $g_d^{(\alpha)}=0$ for $0<\alpha\le d$, since $g$ is monic 
of degree~$d$ in $Y$. The next lemma states that these conditions uniquely determine 
the bihomogeneous components of $g$ if we choose $D\ge d^2$.
The proof is based on well-known ideas from the application of the LLL-algorithm
to polynomial factoring~\cite{lell:82}
(see also~\cite[Lemma 16.20]{gage:99}), adapted from $\ZZ$ to the setting of a polynomial ring. 

\begin{lemma}\label{le:minpol}
By comparing the coefficients of the powers of the indeterminate~$T$, one can interpret the conditions 
$$
 \sum_{i,\alpha \le d} Z_{i,\alpha} T^\alpha (\sum_{\delta\le 2d^2}\varphi^{(\delta)} T^\delta)^i 
   \ \equiv\ 0 \bmod T^{2d^2 +1}, 
$$ 
\vspace{-2ex}
$$
   Z_{d,0}=1, Z_{d,1}=0,\ldots, Z_{d,d}=0
$$
as a system of linear equations over the field $k(X)$ in the unknowns $Z_{i,\alpha}$. 
(There are $2d^2+1$ equations and $(d+1)^2$ unknowns).
This linear system has as the unique solution the bihomogeneous components 
$Z_{i,\alpha} =g_i^{(\alpha)}$ of $g$.
\end{lemma}

\begin{proof}
We define the bivariate polynomial  
$A(T,Y):= \sum_{i,\alpha \le d} g_i^{(\alpha)} T^\alpha Y^i$ 
over $k(X)$ and assign to a solution $\zeta_{i,\alpha}$ of the above linear system of equations 
the bivariate polynomial
$B(T,Y):= \sum_{i,\alpha \le d} \zeta_{i,\alpha} T^\alpha Y^i$. 
Note that $A$ is an irreducible polynomial in $Y$ over $k(X,T)$ since we assume 
$g$ to be irreducible and monic with respect to $Y$. 
The polynomial
$\psi := \sum_{\delta\le 2d^2}\varphi^{(\delta)} T^\delta$ 
is an approximative common root of $A$ and $B$ in the sense that
$$
 A(T,\psi) \equiv 0 \bmod T^{2d^2+1}, \quad
 B(T,\psi) \equiv 0 \bmod T^{2d^2+1} .
$$
The resultant $\mathrm{res}(A,B)\in k(X)[T]$ of $A$ and $B$ with respect to~$Y$ 
satisfies the degree estimate
$$
 \deg_T\mathrm{res}(A,B)) \le  \deg_Y A \cdot \deg_T B +  \deg_T A \cdot \deg_Y  B \le 2d^2, 
$$
which is easily seen from the description of the 
resultant as the determinant of the Sylvester matrix (cf.~\cite[\S 6.3]{gage:99}). 
It is well-known that there exist polynomials
$u,v\in k(X)[T,Y]$ such that $u A + v B  = \mathrm{res}(A,B)$. 
Substituting the approximative common root $\psi$ for $Y$ in this equation implies that 
$\mathrm{res}(A,B)\equiv 0 \bmod T^{2d^2+1}$, hence 
the resultant vanishes. Since $A$ is irreducible, it must be a factor of $B$
over $k(X,T)$. However, we assume $A$ and $B$ both to be monic with respect to $Y$.
This implies that in fact $A=B$ as claimed.
\end{proof}

\medskip

The coefficients of the linear system of equations in Lemma~\ref{le:minpol} can be computed 
from the homogenous components $\varphi^{(\delta)}$, $\delta \le 2d^2$, with $d$ 
multiplications of power series given by their coefficients up to degree $2d^2$.
This can be done with $O(dM(d^2))$ arithmetic operations (Lemma~\ref{le:compos}). 

Assume that $n$ by $n$ matrices can be multiplied with 
$O(n^\gamma)$ arithmetic operations. Then we can compute from 
the coefficients of the linear system the unique solution with 
$O(d^{2\gamma})$ operations (see~\cite[Chap.~16]{bucs:96}). 
This computation can be interpreted as a straight-line program involving divisions.
However, as the bihomogeneous components of~$g$ we are seeking for are homogenous 
of degree at most~$d$, 
we can apply Strassen's idea of avoiding divisions~\cite{stra:73-1} 
and transform this straight-line program into one without divisions, which is
at most by a factor of $O(M(d))$ longer. Summarizing, we obtain the following: 

\begin{equation}\label{pro:min_pol} 
 L\big(\{g_i^{(\delta)} \mid i,\delta \le d\} \mid \varphi^{(1)},\ldots,\varphi^{(2d^2)} \big) 
       \ \le\ O\big(d^{2\gamma} M(d)\big) .
\end{equation}

Our main Theorem~\ref{th:main} is a consequence of this estimate and Proposition~\ref{pro:graph}.
In fact, this provides an upper bound on $\bL(g)$ with respect to the field extension $k_1$ of $k$
of degree at most~$d$ considered at the beginning of Section~\ref{sse:graph}. 
To simplify notation, we assumed there that $k_1=k$. 
This assumption can now be eliminated at the price of an additional factor $M(d)$ 
in the complexity bound according to Lemma~\ref{le:aux}(2). 
As noted in the proof, we may directly take $k_1=k$ in the cases $k=\RR$ or $\CC$, 
so that this additional factor is not necessary in these cases. 
Moreover, note that if $g$ is the generator of the graph of a polynomial function, 
we obtain the improved bound stated in Remark~\ref{re:main-supp} directly from
Proposition~\ref{pro:graph}. 
Summarizing, we have now provided the proof of the main Theorem~\ref{th:main} 
as well as of Remark~\ref{re:main-supp}. 

\begin{remark}
Alternatively, one can compute the bihomogeneous components of~$g$ by an 
analogue of the LLL-algorithm applied to the lattice 
$$
\{A\in k(X)[T,Y]\mid \deg_Y A \le d, A(T,\psi)\equiv 0 \bmod T^{2d^2+1} \} 
      \simeq R + RY + \cdots R^d 
$$
over the principal ideal domain $R:= k(X)[T]$.  
The complexity bound resulting from this approach is $O(d^4\cdot d^2)$ operations in $R$
(cf.~\cite[Thm.~4.8]{gath:84} or \cite{lens:85}). 
This results in $O(d^6 M(d))$ operations in $k(X)$, including divisions, 
which is worse than the bound $O(d^{2\gamma})$ 
of Proposition~\ref{pro:min_pol} for Gaussian elimination ($\gamma=3$).

We think that an improvement upon the bound of inequality~(\ref{pro:min_pol}) is possible 
by taking account of the structure of the linear system of equations under consideration, 
based on the ideas of Wiedemann~\cite{wied:85,kasa:91}. 
This reduces to the question of how fast the matrix underlying the above linear system 
can be multiplied with a vector.
This is an interesting problem in its own right, which will be addressed elsewhere. 
\end{remark}

%%%%%%%%%%%%%%%%%%%%%%%%%%%%%%%%%%%%%%%%%%%%%%%%%%%%%%%%%%%%%%%%%%%%%%%%%%%%%%%%%%%%%%
\section{Applications to Decision Complexity}\label{se:appl}

By combining Theorem~\ref{th:main}, Remark~\ref{re:main-supp}, and Lemma~\ref{le:de-comp}, 
we obtain the following corollary.

\begin{corollary}\label{cor:de-comp}
Let $g$ be the generator of an irreducible hypersurface in $\RR^n$ or $\CC^n$ of degree~$d$. 
Assume that $n$ by $n$ matrices can be multiplied with $O(n^\gamma)$ arithmetic operations.
Then we have 
$$
 \bL(g) \ \le\ O\big( M(d^4)C(g) + d^{2\gamma} M(d) \big) .
$$
\end{corollary}

We remark that if the hypersurface is the graph of a polynomial function, then 
we obtain the better bound $\bL(g) \le O\big(M(d^2)C(g)\big)$. 
This implies Corollary~\ref{cor:main} of the introduction for deterministic decision complexity. 
The claim about randomized complexity (formalized by randomized algebraic computation trees) 
then follows easily by the results in~\cite{meye:85,bukl:93,ckklw:95}.

In~\cite{vali:79-3,vali:82} Valiant had proposed an analogue of the theory of 
$\NP$-completeness in a framework of algebraic complexity, in connection with his famous 
hardness result for the permanent~\cite{vali:79-2}. 
This theory features algebraic complexity classes $\VP$ and $\VNP$ as well as $\VNP$-completeness 
results for many families of generating functions of graph properties, the most prominent 
being the family of permanents.
There is rather strong evidence for Valiant's hypothesis $\VP\ne\VNP$. In fact, if 
it were false, then the nonuniform versions of the complexity classes $\NC$ and $\PH$ would
collapse~\cite{buer:97-1}.  
For a comprehensive presentation of this theory, we refer 
to~\cite{gath:87-1,bucs:96,buer:00-3}. 
In the following, we assume some basic familiarity with the concepts introduced there. 

It is quite natural to incorporate the concept of approximative complexity into 
Valiant's framework. 

\begin{definition}\label{def:bVP}
An {\em approximatively $p$-computable family} is a $p$-family $(f_n)$ 
such that $\bL(f_n)$ is a $p$-bounded function of~$n$. The complexity class 
$\bVP$ comprises all such families over a fixed field $k$. 
\end{definition}

It is obvious that $\VP\sus\bVP$. If the polynomial $f$ is a projection of a 
polynomial~$g$, then we clearly have $\bL(f)\le\bL(g)$. Therefore, 
the complexity class $\bVP$ is closed under $p$-projections. 
We remark that $\bVP$ is also closed under the polynomial oracle reductions 
introduced in~\cite{buer:99-1}. 

We know very few about the relationship between the 
complexity classes $\VP$, $\bVP$, and $\VNP$. 
We therefore raise the following question:

\begin{problem}\label{prob:vp-bvp}
Is the class $\VP$ strictly contained $\bVP$? 
\end{problem}

Intuitively, one would think that $\bVP$ should not differ too much from $\VP$. 
Commenting on this, we remark that 
by Lemma~\ref{le:bLv}(3), an improvement of Theorem~\ref{th:OA} of the form 
$\max\{q,\bL_q(f)\} \le (\bL(f) +\deg f)^{O(1)}$ would imply that $\VP=\bVP$. 
However, we do not see how to achieve such an improvement. 

The class $\VNP$ is closed under taking coefficients (cf.~\cite[\S 2.3]{buer:00-3}). 
This makes it plausible that $\bVP$ is contained in $\VNP$. Nevertheless, this is not 
clear as the occuring polynomials might have a degree exponential in~$\epsilon$. 

The hypothesis 
\begin{equation}\label{hypo}
\VNP\not\subseteq\bVP
\end{equation}
is a strengthening of Valiant's hypothesis, which 
is equivalent to saying that $\VNP$-complete families 
are not approximately $p$-compu\-ta\-ble. 

This hypothesis should be compared with the known work on polynomial time deterministic 
or randomized approximation algorithms for the permanent of 
non-negative matrices~\cite{lisw:98,barv:99,jesv:00}. 
Based on the Markov chain approach, Jerrum, Sinclair and Vigoda~\cite{jesv:00} 
have recently established a fully-polynomial randomized approximation scheme for computing 
the permanent of an arbitrary real matrix with non-negative entries. 
We note that this result does not contradict hypothesis~(\ref{hypo}),  
since the above mentioned algorithm works only for matrices with 
{\em non-negative} entries, while approximative straight-line programs 
a fortiori work on all real inputs.

Under the hypothesis $\VNP\not\subseteq\bVP$, we can conclude 
from Corollary~\ref{cor:de-comp} that 
checking the values of  polynomials forming $\VNP$-complete families 
is hard, even when we allow randomized algorithms with two-sided error.

\begin{corollary}\label{cor:test_compl_fam}
Assume $\VNP\not\subseteq\bVP$ over a field~$k$ of characteristic zero. 
Then for any $\VNP$-complete family $(g_n)$, checking the value $y=g_n(x)$ 
cannot be done by deterministic or randomized algebraic computation trees 
with a polynomial number of arithmetic operations and tests in~$n$. 
\end{corollary}

Hypothesis~(\ref{hypo}) implies a separation of complexity classes in the 
Blum-Shub-Smale model of computation~\cite{blss:89}.
See~\cite{bcss:95} for the definition of the classes $\Po_\RR$ and $\Par_\RR$. 
(For the proof use Corollary~\ref{cor:de-comp} with the permanent polynomial~$g$.)  

\begin{corollary}\label{cor:va->bss}
If\/ $\VNP\not\subseteq\bVP$ is true, then we have 
$\Po_\RR \ne \Par_\RR$ in the Blum-Shub-Smale model over the reals. 
\end{corollary}

%%%%%%%%%%%%%%%%%%%%%%%%%%%%%%%%%%%%%%%%%%%%%%%%%%%%%%%%%%%%%%%%%%%%%%%%%%%%%%%%%%%%%%
\section{Properties of Approximative Complexity}\label{se:approx-supp}

We perform here a more detailed analysis of the concept of approximative complexity. 
The results of this section are not needed for understanding the main results of the paper. 
The field~$k$ may here also be of positive characteristic. 

%%%%%%%%%%%%%%%%%%%%%%%%%%%%%%%%%%%%%%%%%%%%%%%%%%%%%%%%%%%%%%%%%%%%%%%%%%%
\subsection{Trailing $\bf p$-adic Coefficients}\label{sse:p-adic-coeff} 

We discuss first a result about the complexity to compute the $p$-adic expansion of a 
polynomial, which is related to Proposition~\ref{pro:comp_hom_parts} and proved in a 
similar way.

Let $A$ be a commutative algebra over the field $k$ and 
$p\in A[Y]$ be a fixed monic polynomial of degree~$d\ge 1$. 
Any polynomial $f\in A[Y]$ has a unique $p$-adic expansion
$f = \sum_{i\ge 0} f_i p^i$, where 
$f_i = \sum_{\mu <d} f_{i,\mu} Y^\mu \in A[T]$ is of degree strictly less than~$d$. 
We will write
$$
 C^p_D(f) :=\{f_{i,\mu} \mid i \le D, \mu <d\} \subseteq A
$$
for the set of coefficients of the $p$-adic coefficients of $f$ up to order $D$.

\begin{lemma}\label{le:p-ad-mult}
For $f,g\in A[T]$ we have 
$$
 L(C^p_D(f\cdot g)\mid C^p_D(f,g)) \ \le\ O(M(Dd)) .
$$
\end{lemma}

\begin{proof}
Let $f=\sum_{i,\mu} f_{i,\mu} Y^\mu p^i$ and $g=\sum_{j,\nu}g_{j,\nu} Y^\nu p^j$ be the $p$-adic 
expansions of $f$ and $g$. Assume that
$$
 \sum_{\ell\le 2D,\lambda <2d-1} h_{\ell,\lambda} Y^\lambda U^\ell = \
  \bigg(\sum_{i\le D,\mu <d} f_{i,\mu} Y^\mu U^i\bigg) 
  \bigg(\sum_{j\le D,\nu <d} g_{j,\nu} Y^\nu U^j \bigg) ,
$$
where $U$ is a new indeterminate. 
The coefficients $h_{\ell,\lambda}$ can be computed by bivariate polynomial multiplication 
with~$O(M(Dd))$ operations. Put 
$h_\ell := \sum_{\lambda < 2d -1}h_{\ell,\lambda} Y^\lambda$. 

Suppose that $\sum_{\ell\le 2D+1}\overline{h}_\ell p^\ell$ is the $p$-adic 
expansion of $f\cdot g = \sum_{\ell\le 2D} h_\ell p^\ell$. It is easy to see that 
the $\overline{h}_0,\ldots,\overline{h}_{2D+1}$ can be obtained from the 
$h_0,\ldots,h_{2D}$ by $2D$ divisions with remainder by~$p$.
Such a division with remainder can be performed with $M(d)$ arithmetic operations in~$A$ 
(cf.~\cite[Cor.~2.26]{bucs:96}). Therefore, the coefficients of 
$\overline{h}_0,\ldots,\overline{h}_{2D+1}$ can be obtained from the 
coefficients of $h_0,\ldots,h_{2D}$ with $O(D M(d))$ arithmetic operations in~$A$.
\end{proof}

The next proposition shows that the computation of the (coefficients of the)
$p$-adic coefficients of a polynomial up to a certain order~$D$ is not much harder 
than the computation of the polynomial. 

\begin{proposition}\label{pro:p-ad} 
For $D\ge 1$ we have 
$$
 L(C^p_D(f)) \ \le\ O(M(Dd)\, L(f)) .
$$
\end{proposition}

\begin{proof}
Let $g_1,\ldots,g_r\in A[Y]$ be the sequence of intermediate results of a 
computation of $f$. 
Suppose that $g_\rho = g_i\cdot g_j$, $i,j < \rho$. By 
Lemma~\ref{le:p-ad-mult}, we can compute the elements of 
$C^p_D(g_\rho)$ from the elements of $C^p_D(g_i,g_j)$ using 
$O(M(Dd))$ arithmetic operations in~$A$. 
If $g_\rho = g_i \pm g_j$, then we can clearly do this with $O(Dd)$ operations. 
In this way, we can successively compute the elements in 
$C^p_D(g_1),\ldots,C^p_D(g_r)$ with the required number of arithmetic 
operations in $A$. 
\end{proof}

We note that the statement of Proposition~\ref{pro:p-ad} does also hold 
for approximative complexity~$\bL$. (The proof is obvious.) 
We remark that Proposition~\ref{pro:comp_hom_parts} of Section~\ref{se:aux} may 
be derived from the above Proposition~\ref{pro:p-ad} by applying to $F(X,Y)$ the 
substitution $X_i\mapsto Y X_i$, $Y\mapsto p:=Y^{d_X+1}$, where 
$d_X=\deg_X F$, and by taking $A=k[X_1,\ldots,X_n]$.
(Of course, it can also be derived directly.) 

Our initial motivation for the introduction of approximative complexity 
was the study of trailing coefficients. We come now back to this issue  
in a more general setting.  

In the following let $A=k[X_1,\ldots,X_n]$ and $p\in A[Y]$ be monic of degree~$d\ge 1$. 
Let $f=\sum_i f_i p^i$ be the $p$-adic expansion of $f\in A[Y]$.
By Proposition~\ref{pro:p-ad} we know that the complexity of the $p$-adic 
coefficient polynomial~$f_i$ of $Y^i$ is polynomially bounded in $d$, $i$, and $L(f)$.
The following proposition essentially going back to Valiant~\cite{vali:82} 
shows that the dependence on the degree~$i$ cannot be avoided in general. 

\begin{proposition}\label{pro:coeff-p-ad}
The complexity of the coefficient polynomials in the $p$-adic expansion 
of a polynomial~$f$ with respect to a polynomial~$p$ 
is not polynomially bounded in $L(f)$ and $\deg p$, 
unless Valiant's hypothesis is false. 
\end{proposition}

\begin{proof} 
We take $p=Y$ and consider the $Y$-adic expansion of the following polynomial $f_n$ 
of complexity $L(f_n) \le O(n^2)$:  
$$
 f_{n} \ :=\ \prod_{i=1}^n \big( \sum_{j=1}^n X_{ij} Y^{2^{j-1}} \big) 
       \ =\ \sum_i f_{n,i} (X) Y^i   .
$$
The coefficient $f_{n,2^n-1}(X)$ equals the sum over all products 
$X_{1 j_1}\cdots X_{n j_n}$ such that $\{j_1,\ldots,j_n\} = \{1,2,\ldots,n\}$. 
That is, $f_{n,2^n-1}(X)$ equals the permanent $\PER_n(X)$ of the matrix $[X_{ij}]$.
An estimate as claimed in the proposition would imply 
that $L(\PER_n(X)) \le n^{O(1)}$, which contradicts Valiant's hypothesis. 
\end{proof}

Assume now that the $p$-adic expansion 
$f= f_q p^q + f_{q+1} p^{q+1} + \ldots$ 
starts at order $q$ ($f_q\ne 0$). 
We call $f_q$ the {\em trailing coefficient} of $f$ with respect to the 
base~$p$. By contrast with Proposition~\ref{pro:coeff-p-ad}, we can say the 
following about the approximative complexity of the trailing coefficient in 
relation to the complexity of $f$. 

\begin{proposition}\label{pro:lead-cof}
The approximative complexity of the trailing coefficient~$f_q$ 
with respect to~$p$
is polynomially bounded in $d=\deg p$ and $\bL(f)$; we have 
$$
 \bL(f_q) \ \le\ O(M(d) \bL(f))   .
$$
\end{proposition}

\begin{proof}
By the semicontinuity of $\bL$ (Lemma~\ref{le:aux}(1)) it 
is sufficient to prove the statement for $L(f)$ on the right-hand side.  
Let $K=k(\epsilon)$ and $R$ be as usual. 
We have $f=f_q p^q + u p^{q+1}$ with some $u\in k[X,Y]$, hence
$f \equiv \epsilon^q f_q + \epsilon^{q+1}u \bmod (p -\epsilon)$.
Let $\rho(u) \in k[X,Y]$ denote the remainder of $u$ by division
with $p- \epsilon$ (viewed as a polynomial in $Y$). 
Then we conclude that 
$\rho(f) = \epsilon^q(f_q + \epsilon \rho(u)) $.
From the definition of approximative complexity we obtain 
$$
 \bL(f_q) \le L(\epsilon^{-q} \rho(f) ) \le 1 + L(\rho(f)). 
$$
On the other hand, we conclude from Proposition~\ref{pro:p-ad}
that $\rho(f)$ can be computed with $O(M(d)L(f))$ arithmetic operations. 
This proves the claim.
\end{proof}

Note that the main reason for us to work with approximative complexity 
is that we do not know whether a statement similar to Prop.~\ref{pro:lead-cof} 
does hold for complexity (compare Problem~\ref{prob:vp-bvp}). 

%%%%%%%%%%%%%%%%%%%%%%%%%%%%%%%%%%%%%%%%%%%%%%%%%%%%%%%%%%%%%%%%%%%%%%%%%%%
\subsection{Further Characterizations}\label{sse:bL-char}

In order to investigate the relationship between $\bL$ and $L$, it is useful to 
introduce a variant $\bL_\infty$ of approximative complexity, which differs from 
$\bL$ at most by a factor of two. 

\begin{definition}\label{def:bLvar}
The {\em approximative complexity $\bL_q(f)$ of order $q\in\NN$} of a polynomial~$f$ 
in $k[X_1,\ldots,X_n]$  is the smallest natural number $r$ such that there exists 
$f'\in k[[\epsilon]][X_1,\ldots,X_n]$ satisfying 
$$
 L(\epsilon^{q} f + \epsilon^{q+1} f' ) \le r ,
$$
where $L$ refers here to the total (division-free) complexity in the 
polynomial ring $k[[\epsilon]][X_1,\ldots,X_n]$ with free constants in 
the ring of formal power series $k[[\epsilon]]$.
Moreover, we define the {\em modified approximative complexity}
$\bL_\infty(f) := \min_q \bL_q(f)$. 
\end{definition}

The following lemma summarizes some of the basic properties of this notion.
The field $k((\epsilon))$ of formal Laurent series is defined as 
the quotient field of $k[[\epsilon]]$. 

\begin{lemma}\label{le:bLv}
\begin{enumerate}
\item[{\rm (1)}] We have $\frac{1}{2}\bL_\infty (f) \le \bL(f) \le \bL_\infty(f) + 1$.  

\item[{\rm (2)}] In Definition~{\rm \ref{def:bLvar}} one can equivalently work with 
the polynomial ring $k[\epsilon]$ instead of with the coefficient ring $k[[\epsilon]]$ 
of formal power series. 
In Definition~{\rm \ref{def:bL}} one can equivalently work with 
$R=k[[\epsilon]]$ and $K=k((\epsilon))$. 

\item[{\rm (3)}] We have $L(f) \le O(M(q)\,\bL_q(f))$. 

\end{enumerate}
\end{lemma}

\begin{proof}
(1)\  For proving the right-hand estimate of~(1), we assume that we have an optimal 
straight-line program of length~$r=\bL_\infty(f)=\bL_q(f)$ computing $\epsilon^{q} f + \epsilon^{q+1} f'$ 
in $k[[\epsilon]][X]$. We can execute this straight-line program in 
$k[[\epsilon]][X]/(\epsilon^{q+1}) \simeq k[\epsilon][X]/(\epsilon^{q+1})$  
by applying the canonical projection. By interpreting this computation back 
in $k[\epsilon][X]$ we obtain that 
$L(\epsilon^{q} f + \epsilon^{q+1} f'')\le r$ for some suitable 
$f''\in k[\epsilon][X]$, where $L$ refers here to $k[\epsilon][X]$. 
We multiply the result with the constant $\epsilon^{-q}$ (this is the only 
computational step leading outside $k[\epsilon][X]$) and conclude that $\bL(f) \le r + 1$. 

For proving the left-hand estimate of~(1), we use the embedding of 
$K=k(\epsilon)$ in the field $k((\epsilon))$ of formal Laurent series. 
This leads to an embedding $K[X] \hookrightarrow k((\epsilon))[X]$. 
The elements of $k((\epsilon))[X]$ can be written in the form
$\epsilon^{-\alpha} \cdot A$ with $\alpha\in\NN$ and $A\in k[[\epsilon]][X]$. 
Note that for $\alpha \ge \beta$
$$
 \epsilon^{-\alpha} A \pm \epsilon^{-\beta} B 
   = \epsilon^{-\alpha} (A \pm \epsilon^{\alpha-\beta}\cdot B), \quad
 \epsilon^{-\alpha} A \cdot \epsilon^{-\beta} B  = \epsilon^{-\alpha-\beta} (A \cdot B) .
$$
If we encode an element $\epsilon^{-\alpha} \cdot A$ by the constant $\epsilon^{-\alpha}$ and the 
polynomial~$A$ over the ring $k[[\epsilon]]$, then we can simulate any 
division-free straight-line computation in $k((\epsilon))[X]$ 
of length~$r$ by a straight-line computation in $k[[\epsilon]][X]$ of length at most~$2r$.
(The number of nonscalar multiplications even remains the same.)
This way, a computation of $F=f+\epsilon f'$ in $k(\epsilon)[X]$ 
with $f'\in k[[\epsilon]][X]$ will lead to a computation in $k[[\epsilon]][X]$ of some $C$ 
such that $\epsilon^{-\gamma}C =F$ for some $\gamma\in\NN$, hence
$C= \epsilon^{\gamma} f + \epsilon^{\gamma+1}f'$.

(2)\ This follows from the first part of the proof of part~(1). 

(3)\ This is a consequence of part~(2) and Proposition~\ref{pro:p-ad} applied to compute $\epsilon$-adic coefficients. 
\end{proof}

Part~(3) of the above lemma provides a polynomial bound on the complexity in terms of the 
approximative complexity of a certain order of approximation~$q$ and this order~$q$.  
Unfortunately, the best general upper bound on the order~$q$, that we are 
able to prove, is exponential in the complexity. 

\begin{theorem}\label{th:OA}
For polynomials $f$ over an algebraically closed field~$k$ we have 
$\bL_{q}(f) \le 2 \bL(f)$ with $q\le 2^{\bL(f)^2}$.
%For polynomials $f$ over an algebraically closed field~$k$ we have 
%$\bL_{q}(f) \le 2 \bL(f)$ with $q\le 4^{\bL(f)}$.
\end{theorem}

%%%
\begin{proof}
We proceed as in Lehmkuhl and Lickteig~\cite{leli:89}, who proved a similar bound on the 
order of approximation for border rank (approximative bilinear complexity). 

The proof is based on the following geometric description of the set 
$\{f\in A_n \mid L(f) \le r\}$. 
The field~$k$ is assumed to be algebraically closed. 
A straight-line program~$\Gamma$ is a description for a computation 
of a polynomial from constants $z_1,\ldots,z_m$ and variables 
$X_1,\ldots,X_n$ (recall that we do not allow divisions). 
Let $\phi_\Gamma(z)$ denote the polynomial in $A_n:=k[X_1,\ldots,X_n]$ 
computed by~$\Gamma$ from the list of constants $z\in k^m$.
Let $r_\ast$ denote the number of multiplication instructions of~$\Gamma$.
Then we have 
$$
 \phi_\Gamma(z) = \sum_\mu  \phi_{\Gamma,\mu}(z) X^\mu ,
$$ 
where the sum runs over all $\mu\in\NN^n$ with 
$\mu_1+\ldots + \mu_n \le 2^{r_\ast}$. 
Moreover, the coefficient polynomials $\phi_{\Gamma,\mu}(z)$ 
have degree at most $2^{r_\ast}$. 
We interpret $\phi_{\Gamma}$ as 
a morphism $k^m \to \{f\in A_n \mid \deg f \le 2^{r_\ast}\}$ of affine varieties.  
Applying \cite[Theorem~8.48]{bucs:96} to the polynomial map 
$z\mapsto (z,\phi_\Gamma(z))$, we see that 
$\deg \graph (\phi_\Gamma) \le \big(2^{r_\ast}\big)^m =: D$. 
The image ${\cal C}_\Gamma$ of $\phi_\Gamma$ is an irreducible, constructible set.
%By B\'ezout's inequality, the geometric degree of the graph of $\phi_\Gamma$ satisfies 
%$\deg \graph (\phi_\Gamma) \le 2^{r_\ast}$ (compare~\cite[\S 8.3]{bucs:96}). 
We have for fixed~$r$ that
$$
\{f\in A_n \mid L(f) \le r\} = \bigcup_{\Gamma} {\cal C}_\Gamma ,
$$
where the union is over all straight-line programs~$\Gamma$ of length~$r$. 

Assume now that $f$ is in the Zariski-closure of the set on the left-hand side. 
Then we have $f\in \overline{{\cal C}_\Gamma}$ for some $\Gamma$. 
(We remark that in the case $k=\CC$ the Zariski-closure of constructible sets 
coincides with the closure with respect to the Euclidean topology 
(cf.~\cite[Theorem~2.33]{mumf:76}). 

We apply now two results proven in Lehmkuhl and Lickteig~\cite{leli:89} 
to the morphism~$\phi_\Gamma$.
Proposition~1 of~\cite{leli:89} claims that there exists an irreducible curve $C\subseteq k^m$ 
such that $f\in\overline{\phi_\Gamma (C)}$ and $\deg C \le \deg \graph (\phi_\Gamma)$. 
The Corollary to Proposition~3 in~\cite{leli:89} states that there exists a point 
$\zeta=(\zeta_1,\ldots,\zeta_m)\in k((\epsilon))^m$ such that $F:=\phi_\Gamma(\zeta)$ is defined over 
$k[[\epsilon]]$, satisfies $F_{\epsilon=0} =f$ and such that all formal 
Laurent series~$\zeta_i$ have order at least $-\deg C$. 
We conclude with Lemma~\ref{le:bLv}(2)
that $L(F)\le r$ and hence $\bL(f)\le r$, which  
proves the nontrivial direction of Theorem~2.4.
Moreover, we have shown that there is a straight-line program of length~$r$, which computes $F$ in
$k((\epsilon))[X]$ from the $X$-variables and constants $\zeta_i$ having order
at least $-\deg C \ge - D$. %-2^{mr_\ast}$. 
By a similar reasoning as in the proof of Lemma~\ref{le:bLv}(1), 
we can construct from 
this a straight-line program of length at most~$2r$, which computes in $k[[\epsilon]][X]$  
an element of the form  $\epsilon^{q} f + \epsilon^{q+1} f'$ with 
$q \le  2^{r_\ast} D = 2^{(m+1)r_\ast}$.
We therefore have $\bL_q(f) \le 2 r$. 
To complete the proof, we note that 
$(m+1)r_\ast \le r^2$, unless $m=r$ and $r=r_\ast$. 
However, in this case, the components of $\phi_\Gamma$ have degree at most~$1$
and we get $q \le  2^{r_\ast}$ since $\deg\graph\phi_\Gamma \le 1$.  
\end{proof}

By tracing the proofs of the above results it is straightforward to show the following statement.

\begin{remark}
By counting only nonscalar multiplications, one can introduce 
the notions $\bL^{\mathit{ns}}, \bL_q^{\mathit{ns}}$ in an analogous way.
We then have $\bL^{\mathit{ns}} = \bL_\infty^{\mathit{ns}} = \bL_q^{\mathit{ns}}$. 
%with $q\le 4^{\bL^{\mathit{ns}}}$.
\end{remark}

Finally, we show that the restriction to division-free straight-line programs 
in the definition of approximative complexity is not a serious one. 

\begin{lemma}\label{le:bL-div}
If $\bL'(f)$ denotes the approximative complexity of a multivariate polynomial~$f$ of degree~$d$, 
where divisons are allowed, then the divison-free approximative complexity~$\bL(f)$ 
can be bounded by $\bL(f) \le O(M(d)\bL'(f))$. 
Here the ground field~$k$ is assumed to be infinite. 
\end{lemma}

\begin{proof}
Formally, $\bL'(f)\le r$ means there exists $F\in R(X)$ with $F_{\epsilon=0}=f$ and such that
the complexity of $F$ in $K(X)$ (allowing divisions) is at most~$r$. 
We can avoid the divisions using the well-known ideas of~\cite{stra:73-1}.
Accordingly, there exists $\xi\in k^n$ such that an optimal computation of $F$ takes places in the local subring 
${\cal O}_\xi$ of $K(X)$ consisting of the rational functions defined at $X=\xi$. To simplify notion, 
we assume without loss of generality that $\xi=0$. 
Let $F^{(\delta)}$ denote the homogeneous component of~$F$ of degree~$\delta$.
By~\cite[Theorem~7.1]{bucs:96}, the division-free complexity $L$ of these components of degree up to~$d$ satisfies
$L(F^{(0)},\ldots,F^{(d)}) =O(M(d)r)$. 
As $F$ is defined over $R$, also its homogeneous components are defined over~$R$ and we have 
$(F^{(\delta)})_{\epsilon=0} =f^{(\delta)}$. This implies that 
$\bL(f^{(0)},\ldots,f^{(d)}) =O(M(d)r)$ as claimed.
\end{proof}

%%%%%%%%%%%%%%%%%%%%%%%%%%%%%%%%%%%%%%%%%%%%%%%%%%%%%%%%%%%%%%%%%%%%%%%%%%%
\section{Appendix}

We include here the proofs of Theorem~\ref{th:fa-I} and 
Proposition~\ref{pro:root}.
Although being essentially the same as the original proofs in~\cite{kalt:86-1}
(with minor improvements in complexity), 
we believe that our exposition, integrated in the coherent framework of this paper, 
will facilitate the reader's understanding 
of the difficulties encountered in extending these results to Theorem~\ref{th:main}. 

\proofof{Theorem~\ref{th:fa-I}} 
Let $f=g^e h$ with coprime polynomials $g,h\in k[X_1,\ldots,X_n]$  
and $d =\deg g$.
As in the proof of Lemma~\ref{le:fe} we may achieve by a linear 
coordinate transformation that $g$~is monic of degree~$d$ with respect to 
the variable $Y:=X_n$. We put $X:=(X_1,\ldots,X_{n-1})$
and $A=k[X]$. 
Using resultants, we see that there is some point $\xi\in k^{n-1}$ 
such that the univariate polynomials $g^{(0)}:=g(\xi,Y)$ and 
$h^{(0)}:=h(\xi,Y)$ are coprime. By a coordinate translation, 
we may assume that $\xi=0$. 
Note that $g^{(0)}$ is a monic univariate polynomial over $k$ of degree~$d$. 
We have $u h^{(0)} + v g^{(0)} = 1$ with uniquely determined $u,v\in k[Y]$ 
such that $\deg u < d$. 

The basic idea is to use Hensel lifting in order to succes\-sively compute 
the factorization $f=g^e h$ from $f^{(0)} := (g^{(0)})^e h^{(0)}$. 
The crux is the choice of the suitable valuation by which to lift. 
We will lift with respect to the total degree in the $X$-variables. 
Consider the decomposition of polynomials in $A[Y]$ into homogeneous parts 
with respect to the total degree in the $X$-variables:
$$
 f = \sum_{\delta\ge 0} f^{(\delta)} , \  
 g = \sum_{\delta\ge 0} g^{(\delta)} , \ 
 h = \sum_{\delta\ge 0} h^{(\delta)}  ,  
$$
where $f^{(\delta)}, g^{(\delta)}, h^{(\delta)}\in A[Y]$ are homogeneous of 
degree $\delta$ in the $X$-variables.
This notation is consistent with our earlier introduction of 
$f^{(0)}, g^{(0)}, h^{(0)}$. 
We have $\deg_Y g^{(\delta)} < d$ for $\delta>0$, and $g^{(\delta)}=0$ for $i>d$. 

We are going to derive a formula, which allows to compute 
$g^{(s+1)}, h^{(s+1)}$ from the homogeneous parts of $g$ and $h$ up to degree~$s$.
From $f=g^e h$ we obtain modulo the ideal generated by the monomials of degree~$s+2$ in 
the $X$-variables that 
\begin{eqnarray*}
f &\equiv& \bigg( \sum_{\delta=0}^s g^{(\delta)} +  g^{(s+1)} \bigg)^e 
           \bigg( \sum_{\delta=0}^s h^{(\delta)} +  h^{(s+1)} \bigg) \\
  &\equiv& \bigg( \bigg(\sum_{\delta=0}^s g^{(\delta)} \bigg)^e + e (g^{(0)})^{e-1} g^{(s+1)} \bigg)
           \bigg( \sum_{\delta=0}^s h^{(\delta)} + h^{(s+1)} \bigg) \\
  &\equiv& \bigg(\sum_{\delta=0}^s g^{(\delta)} \bigg)^e \bigg(\sum_{\delta=0}^s h^{(\delta)} \bigg) 
             + e (g^{(0)})^{e-1} g^{(s+1)} h^{(0)} + (g^{(0)})^e h^{(s+1)} .
\end{eqnarray*}
If we write 
$$
 F:= \sum_{\delta\ge 0} F^{(\delta)} := \bigg( \sum_{\delta=0}^s g^{(\delta)} \bigg)^e 
                       \bigg( \sum_{\delta=0}^s h^{(\delta)} \bigg)  
$$
(omitting the dependence of $F$ on~$s$) and set
$\Delta^{(s+1)} := f^{(s+1)} - F^{(s+1)}$, then we obtain 
\begin{equation}\label{eq:hensel-rec}
  Q:=\frac{\Delta^{(s+1)}}{(g^{(0)})^{e-1}} = e g^{(s+1)} h^{(0)} + g^{(0)} h^{(s+1)} .   
\end{equation} 
Since $\deg_Y g^{(s+1)}<d$, this relation uniquely determines the polynomial $g^{(s+1)}$. 
On the other hand, we have $Q=uQ h^{(0)} + vQ g^{(0)}$. 
It follows that $eg^{(s+1)}$ is the remainder of the division of $u Q$ by $g^{(0)}$, 
hence $eg^{(s+1)}$ is the $(e-1)$th $g^{(0)}$-adic coefficient of $u \Delta^{(s+1)}$. 

We write elements $a\in A[Y]$ in the form
$$
 a = \sum_{\delta\ge 0} a^{(\delta)} = \sum_{\delta\ge 0, i\ge 0} a_i^{(\delta)} (g^{(0)})^i 
   = \sum_{\delta\ge 0, i\ge 0, j<d} a_{ij}^{(\delta)} Y^j (g^{(0)})^i 
$$
with {\em coefficients} $a_{ij}^{(\delta)}\in A$, which are homogeneous polynomials of 
degree~$\delta$ in the $X$-variables. Note that the 
$a_i^{(\delta)} = \sum_{j<d} a_{ij}^{(\delta)} Y^j$ 
are the $g^{(0)}$-adic coefficients of $a^{(\delta)}$. 
The collection of coefficients
$$
  \mbox{$a_{ij}^{(\delta)}$\quad for $0\le \delta\le d,\ 0\le j < d,\ 0 \le i \le D$}
$$
will be used to represent the element $a$ approximatively. 
We will call this an approximation of order $D$ of $a$. 

Assume that $c= a\cdot b$ in $A[Y]$. A straightforward generalization of 
Lemma~\ref{le:p-ad-mult} yields that the coefficients $c_{ij}^{(\delta)}$ 
of $c$ with $\delta\le d, j<d, i\le D$ can be computed from the corresponding 
coefficients of $a$ and $b$ with $O(M(d^2 D))$ arithmetic operations in $A$. 
Using this generalization of Lemma~\ref{le:p-ad-mult}, we can generalize 
Proposition~\ref{pro:p-ad} in an obvious way and obtain that the coefficients 
$f_{ij}^{(\delta)}$ for $\delta \le d, j< d, i\le de$ 
can be computed with  $O(M(d^2 \cdot de)L(f))$ arithmetic operations in $A$.  

For $0\le s\le d$ we define $C_s$ as the set of the coefficients 
$g_{ij}^{(\delta)}, h_{ij}^{(\delta)}$ for  $ \delta\le s, j<d, i \le (d-s)e$. 
Note that $C_s$ is a subset of $A=k[X]$. 
Moreover, $C_0\subseteq k$, thus the  elements of $C_0$ can be considered as free constants. 
Note also that $g^{(\delta)} = g_0^{(\delta)}$ for $\delta >0$ as $\deg_Y g^{(\delta)} <d$. 

Inductively, we assume now that we have already computed the elements of $C_s$ for $s<d$. 
From these data we can compute the coefficients 
$F_{ij}^{(s+1)}$ of $F$ for $j<d, i\le (d-s)e$ with $O(M(d^2(d-s)e\log e))$ 
arithmetic operations using the above mentioned generalization of  Lemma~\ref{le:p-ad-mult} 
($\log e$ squarings). 
From this and the coefficients $f_{ij}^{(\delta)}$, $\delta \le d, j < d, i\le de$, we compute 
the coefficients of $u \Delta^{(s+1)}$ up to order $(d-s)e$ 
with $O(M(d(d-s)e))$ arithmetic operations. 
In particular, we have thus computed the coefficients of $eg^{(s+1)}$, since 
$e g^{(s+1)}$ is the $(e-1)$th $g^{(0)}$-adic coefficient of $u \Delta^{(s+1)}$. 
From equation~(\ref{eq:hensel-rec}) we obtain 
$$
 \Delta^{(s+1)} - e g^{(s+1)} (g^{(0)})^{e-1} h^{(0)} =  h^{(s+1)} (g^{(0)})^e .
$$
We can compute the coefficients of this polynomial up to order $(d-s)e$ 
with further $O(M(d(d-s)e))$ arithmetic operations. 
This way, we get the coefficients of $h^{(s+1)}$ up to order $(d-s-1)e$. 
Note that the order has decreased by $e$.  

Summarizing, the cost of each induction step is $O(M(d^3 e\log e))$ 
and there are at most $d$ induction steps. 
The polynomial $g=g^{(0)} + g^{(1)} + \ldots +  g^{(d)}$ can be computed with further $O(d^2)$ additions
from the coefficients of the $g^{(\delta)} = g_0^{(\delta)}$ for $1\le \delta \le d$. 
Altogether, we obtain $L(g)=O(M(d^3 e)(L(f) + d\log e))$ as claimed.
\proofend

\begin{proposition}\label{pro:root}
Assume that $f=g^e$ in $k[X_1,\ldots,X_n]$, $d=\deg g \ge 1$ and $\chara k =0$. 
Then $L(g) \le O(M(d)L(f))$. 
\end{proposition}

\begin{proof}
By a coordinate shift, we may assume that $g(0)\ne 0$; without loss of generality $g(0)=1$. 
The polynomial $\varphi = g - 1 \in k[X]$ is the unique solution of the equation
$$
 (1+\varphi)^e - f = 0, \ \varphi(0)=0
$$ 
to be solved in the ring of formal power series $k[[X]]$. 
This power series $\varphi$ can be recursively computed by the Newton iteration 
$\varphi_0 =0$ and  
$$
 \varphi_{\nu+1} = \varphi_\nu - \frac{(1+\varphi_\nu)^e -f }{e(1+\varphi_\nu)^{e-1}}
                 = -\frac{1}{e} + (1 - \frac{1}{e})\varphi_\nu + \frac{1}{e} \frac{f}{(1+\varphi_\nu)^{e-1}}
$$
satisfying $\varphi_\nu\equiv \varphi \bmod (X)^{2^\nu}$.
(Compare Section~\ref{sse:graph} and~\cite[Theorem 2.31]{bucs:96}.) 

We first compute the homogeneous parts of~$f$ up to degree~$d$ with $O(M(d)L(f))$ arithmetic operations 
by a variant of Proposition~\ref{pro:comp_hom_parts}.
Using Lemma~\ref{le:compos}, we can compute from this and the homogeneous parts of $\varphi_\nu$ 
up to degree $2^{\nu}$ the homogeneous parts of $\varphi_{\nu+1}$ up to degree $2^{\nu+1}$ 
with  $O(M(2^\nu) \log e)$ arithmetic operations ($\log e$ squarings). 
As $\sum_{\nu=0}^N M(2^\nu)) \le M(2^{N+1} -1)$, a total of 
$O(M(d)(L(f) + \log e))$ arithmetic operations is sufficient. 
Since $L(f) \ge \log\deg f \ge \log e $, the claim follows. 
\end{proof}

%%%%%%%%%%%%%%%%%%%%%%%%%%%%%%%%%%%%%%%%%%%%%%%%%%%%%%%%%%%%%%%
%\bibliography{/fields/user/pbuerg/mypapers/biblio/lit-bank}

%%%%%%%%%%%%%%%%%%%%%%%%%%%%%%%%%%%%%%%%%%%%%%%%%%%%%%%%%%%%%%%
\end{document}